\documentclass[sigconf]{aamas} 

\usepackage{balance} 
\usepackage{amsmath,mathdots}
\usepackage{calligra}
\usepackage{algorithm,algorithmicx,algpseudocode}

\setcopyright{ifaamas}
\acmConference[AAMAS '24]{Proc.\@ of the 23rd International Conference
on Autonomous Agents and Multiagent Systems (AAMAS 2024)}{May 6 -- 10, 2024}
{Auckland, New Zealand}{N.~Alechina, V.~Dignum, M.~Dastani, J.S.~Sichman (eds.)}
\copyrightyear{2024}
\acmYear{2024}
\acmDOI{}
\acmPrice{}
\acmISBN{}

\acmSubmissionID{1097}

\title[AAMAS-2024 Formatting Instructions]{Simple $k$-crashing Plan with a Good Approximation Ratio}

\author{Ruixi Luo}
\affiliation{
  \institution{School of Intelligent Systems Engineering, Sun Yat-Sen University}
  \city{Shenzhen}
  \country{China}}
\email{luorx@mail2.sysu.edu.cn}

\author{Kai Jin}
\affiliation{
	\institution{School of Intelligent Systems Engineering, Sun Yat-Sen University}
	\city{Shenzhen}
	\country{China}}
\email{cscjjk@gmail.com}

\author{Zelin Ye}
\affiliation{
  \institution{School of Intelligent Systems Engineering, Sun Yat-Sen University}
  \city{Shenzhen}
  \country{China}}
\email{zlrelay@outlook.com}

\begin{abstract}
In project management, a project is typically described as an activity-on-edge network (AOE network), where each activity / job is represented as an edge of some network $N$ (which is a DAG). Some jobs must be finished before others can be started, as described by the topology structure of $N$. It is known that job $j_i$ in normal speed would require $b_i$ days to be finished after it is started. Given the network $N$ with the associated edge lengths $b_1,\ldots,b_m$, the duration of the project is determined, which equals the length of the critical path (namely, the longest path) of $N$.

To speed up the project (i.e. reduce the duration), the manager can crash a few jobs (namely, reduce the length of the corresponding edges) by investing extra resources into that job. However, the time for completing $j_i$ has a lower bound due to technological limits -- it requires at least $a_i$ days to be completed. Moreover, it is expensive to buy resources. Given $N$ and an integer $k\geq 1$, the $k$-crashing problem asks the minimum amount of resources required to speed up the project by $k$ days. We show a simple and efficient algorithm with an approximation ratio $\frac{1}{1}+\ldots+\frac{1}{k}$ for this problem.

We also study a related problem called $k$-LIS, in which we are given a sequence $\omega$ of numbers and we aim to find $k$ disjoint increasing subsequence of $\omega$ with the largest total length. We show a $(1-\frac{1}{e})$-approximation algorithm which is simple and efficient.
\end{abstract}


\newcommand{\BibTeX}{\rm B\kern-.05em{\sc i\kern-.025em b}\kern-.08em\TeX}

\begin{document}


\pagestyle{fancy}
\fancyhead{}


\maketitle 


\section{Introduction}

Critical Path Method (CPM) are commonly used in network analysis and project management. A fundamental problem in CPM is to optimize the total project cost subject to a prescribed completion date
\cite{Siemens1971}.
The general setting is that the project manager can expedite a few jobs by investing extra resources, hence reduce the duration of the entire project and thus meet the desired completion date. Meanwhile, the extra resources spent should be as few as possible.
A formal description of the problem is given in subsection~\ref{subsect:formal-description}.
In this paper, we revisit a simple incremental algorithm for solving this problem and prove that it has an approximation ratio $\frac{1}{1}+\ldots+\frac{1}{k}$, where $k$ denotes the amount of days the duration of the project has to be shortened (which is easily determined by the prescribed completion date and the original duration of the project).

The mentioned algorithm is based on greedy. It speeds up the project by only one day at a time, and repeats it $k$ times. Each time it adopts the minimum cost strategy to shorten the project (by one day).
This does not guarantee the minimum total cost for $k>1$ (see an example in subsection~\ref{subsect:counter-example}).
Nevertheless, it is a simple, efficient, and practical algorithm, which should have been applied by several engineers.
Therefore, a theoretical analysis on its approximation performance seems to be important and may benefit relevant researchers and engineers. Such an analysis is not easy and is absent in literature, hence we cover it in this paper.

\smallskip As a comparison, we also conduct an analysis of a similar greedy algorithm for another fundamental problem called \emph{$k$-LIS problem},
in which we are given a sequence $\omega$ and we aim to find $k$ disjoint increasing subsequence of $\omega$
with the largest total length.
The greedy algorithm for this problem is as the following:
Find the longest increasing subsequence (LIS) at one time and remove it. Repeat it for $k$ times. However, this does not guarantee an optimum solution as well; for example, consider $\omega=(3,4,5,8,9,1,6,7,8,9), k = 2$. Nevertheless, we prove that it is a $(1-\frac{1}{e})$-approximation algorithm.
Moreover, we show that for every $k>1$, there are cases for which the greedy solution is only $0.75$ times the maximum solution.

\subsection{Description of the $k$-crashing problem}\label{subsect:formal-description}

A project is considered as an activity-on-edge network (AOE network, which is a directed acyclic graph) $N$, where each activity / job of the project is an edge.
Some jobs must be finished before others can be started, as described by the topology structure of $N$.

It is known that job $j_i$ in normal speed would take $b_i$ days to be finished after it is started, and hence the (normal) duration of the project $N$, denoted by $d(N)$, is determined, which equals the length of the critical path (namely, the longest path) of $N$.

To speed up the project, the manager can crash a few jobs (namely, reduce the length of the corresponding edges) by investing extra resources into that job.
However, the time for completing $j_i$ has a lower bound due to technological limits - it requires at least $a_i$ days to be completed.
Following the convention, assume that the duration of a job has a linear relation with the extra resources put into this job; equivalently, there is a parameter $c_i$ (slope), so that
shortening $j_i$ by $d~(0\leq d \leq b_i-a_i)$ days costs $c_i\cdot d$ resources.

Given project $N$ and an integer $k\geq 1$, the \emph{$k$-crashing problem} asks the minimum cost to speed up the project by $k$ days.

\smallskip In fact, many people also care about the case of non-linear relation, especially the convex case
where shortening an edge becomes more difficult after a previous shortening.
Delightfully, the greedy algorithm performs equally well for this convex case.
Without any change, it still finds a solution with the approximation ratio $\frac{1}{1}+\ldots+\frac{1}{k}$.
For simplicity, we focus on the linear case throughout the paper, but our proofs are still correct for the convex case.

\subsection{Related work (about $k$-crashing)}

The first solution to the $k$-crashing problem was given by Fulkerson \cite{Fulkerson1961}
and by Kelley \cite{Kelley1961} respectively in 1961. The results in these two papers are independent, yet the approaches are essentially the same, as pointed out in \cite{Phillips1977}.
In both of them, the problem is first formulated into a linear program problem, whose dual problem
is a minimum-cost flow problem, which can then be solved efficiently.

Later in 1977, Phillips and Dessouky \cite{Phillips1977} reported another clever approach (denoted by Algorithm~PD). Similar as the greedy algorithm mentioned above, Algorithm~PD also consists of $k$ steps, and
each step it locates a minimal cut in a flow network derived from the original project network.
This minimal cut is then utilized to identify the jobs which should be expedite or de-expedite in order to reduce the project reduction. It is however not clear whether this algorithm can always find an optimal solution. We have a tendency to believe the correctness, yet cannot find a proof in \cite{Phillips1977}.

It is noteworthy to mention that Algorithm~PD shares many common logics with the greedy algorithm we considered. Both of them locates a minimal cut in some flow network and then use it to identify the set of jobs to expedite / de-expedite (there is no de-expedite allowed in the greedy algorithm, but there are de-expedites in Algorithm~PD) in the next round. However, the constructed flow networks are different. The one in the greedy algorithm  has only capacity upperbounds, whereas the one in Algorithm~PD has both capacity upperbounds and lowerbounds and is thus more complex.

The greedy algorithm we considered is much simpler and easier to implement comparing to all the approaches above.

Other approaches for the problem are proposed by Siemens \cite{Siemens1971} and Goyal \cite{GOYAL1996},
but these are heuristic algorithms without any guarantee -- approximation ratio are not proved in these papers.

\smallskip Many variants of the $k$-crashing problem have been studied in the past decades;
see \cite{proj-acc-2008}, \cite{fast-tracking-2019}, and the references within.

\subsection{Related work (about $k$-LIS)}


The $k$-LIS problem is a generalization of the well-known LIS problem,
most of the study about $k$-LIS lies on revealing the connections between LIS and Young tableau.
Utilizing the Young tableau, Schensted~\cite{schensted1961longest} managed to compute the length of the longest increasing and decreasing subsequence of $\omega$, which offers a way to solve the LIS problem.
Greene \cite{1974An} extended the result of \cite{schensted1961longest} to $k$-LIS problem and calculates the exact result (the $k$ increasing subsequences) of the $k$-LIS problem. First, it takes the largest $k$ rows of the Young tableau of $\omega$. Then it transform $\omega$ into a canonical form $\tilde{\omega}$. Next, it transform $\tilde{\omega}$ back to $\omega$ while modifying the largest $k$ rows of the Young tableau accordingly. In the end, we can get the modified $k$ rows as the optimal result of the $k$-LIS problem. The results of \cite{1974An} are applied in some problems related to LIS; see \cite{thomas2011longest}~\cite{ferrari2022preisach}.

Another approach to solve the $k$-LIS problem is as follows.
We first formulate the $k$-LIS problem into a minimum-cost network flow problem
(which contains $O(n^2)$ arcs, where $n$ denotes the length of the given sequence).
Basically, if some element $u$ of $\omega$ is smaller than $v$ of $\omega$ and $u$ is preceding $v$,
we build an arc from node $u'$ to node $v$.
By solving the flow problem, we obtain the optimum solution of the $k$-LIS problem.
However, the time complexity of this approach is much higher comparing to the greedy algorithm.

\smallskip Finding an approximation of the LIS can be solved in sublinear time~\cite{mitzenmacher2021improved}.
See more results about LIS in a monograph~\cite{romik2015surprising}.

\section{$k$-crashing problem}\label{sect:k-crashing}

\newcommand{\cost}{\mathsf{cost}}
\newcommand{\OPT}{\mathsf{OPT}}
\newcommand{\eba}{E_{b\setminus a}}

\begin{description}
	\item [Project $N$.]
	Assume $N=(V,E)$ is a directed acyclic graph with a single source node $s$ and a single sink node $t$ (a source node refers to a node without incoming edge, and a sink node refers to a node without outgoing edges).
	Each edge $e_i\in E$ has three attributes $(a_i,b_i,c_i)$, as introduced in
	subsection~\ref{subsect:formal-description}. \smallskip
	
	\item [Critical paths \& critical edges.]
	A path of $N$ refers to a path of $N$ from source $s$ to sink $t$. Its length is the total length of the edges included, and the length of edge $e_i$ equals $b_i$.
	The path of $N$ with the longest length is called a \emph{critical path}. The duration of $N$ equals the length of the critical paths. There may be more than one critical paths. An edge that belongs to some critical paths is called a \emph{critical edge}. \smallskip
	
	\item [Accelerate plan $X$.] Denote by $\eba$ the multiset of edges of $N$ that contains $e_i$ with a multiplicity $b_i-a_i$.
	Each subset $X$ of $\eba$ (which is also a multiset)
	is called an \emph{accelerate plan}, or \emph{plan} for short.
	The multiplicity of $e_i$ in $X$, denoted by $x_i~(0\leq x_i\leq b_i-a_i)$, describes how much length $j_i$ is shortened;
	i.e., $j_i$ takes $b_i-x_i$ days
	when plan $X$ is applied.
	The \emph{cost} of plan $X$, denoted by $\cost(X)$, is $\sum_i c_ix_i$.
	\smallskip
	
	\item [Accelerated project $N(X)$.] Define $N(X)$ as the project that is the same as $N$, but with $b_i$ decreased by $x_i$; in other words, $N(X)$ stands for the project optimized with plan $X$.
	\smallskip
	
	\item [$k$-crashing.] We say a plan $X$ is \emph{$k$-crashing}, if the duration of the project $N$ is shortened by $k$ when we apply plan $X$. \smallskip
	
	\item[Cut of $N$.] Suppose that $V$ is partitioned into two subsets $S,T$ which respectively contain $s$ and $t$.
	Then, the set of edges from $S$ to $T$ is referred to as a \emph{cut} of $N$. Notice that we cannot reach $t$ from $s$ if any cut of $N$ is removed.
\end{description}

Let $k_{max}$ be the duration of the original project $N$
minus the duration of the accelerated project $N(\eba)$.
Clearly, the duration of $N(X)$ is at least the duration of $N(\eba)$,
since $X$ is a subset of $\eba$.
It follows that a $k$-crashing plan only exists for $k\leq k_{max}$.

Throughout this section, we assume that $k\leq k_{max}$.

\smallskip The greedy algorithm in the following (see Algorithm~\ref{alg:greedy-crashing}) finds a $k$-crashing plan efficiently. It finds the plan incrementally -- each time it reduces the duration of the project by 1.

\begin{algorithm}[hbpt]
	\caption{Greedy algorithm for finding a $k$-crashing plan.}		\label{alg:greedy-crashing}
	\textbf{Input:} A project $N=(V,E)$.\\
	\textbf{Output:} A $k$-crashing plan $G$.
	\begin{algorithmic}
		\State $G\leftarrow \varnothing$;
		\For {$i = 1$ to $k$}
		\State Find the optimum $1$-crashing plan of $N(G)$, denoted by $A_i$;
		\State $G\leftarrow G \cup A_i$; (regard as multiset union)
		\EndFor
	\end{algorithmic}
\end{algorithm}

Observe that $G$ is an $i$-crashing plan of $N$ after the $i$-th iteration $G\leftarrow G\cup A_i$, as the duration of $N(G)$ is reduced by 1 at each iteration. Therefore, $G$ is a $k$-crashing plan at the end.

If a $1$-crashing plan of $N(G)$ does not exist in the $i$-th iteration $i\leq k$, we determine that there is no $k$-crashing plan; i.e. $k>k_{max}$.

\begin{theorem}\label{thm:upper bound}
	Let $G=A_1 \cup \ldots \cup A_k$ be the $k$-crashing plan found by Algorithm~\ref{alg:greedy-crashing}. Let $\OPT$ denote the optimal $k$-crashing plan. Then, $$\cost(G) \leq \sum_{i=1}^{k}\frac{1}{i}\cost(\OPT).$$
\end{theorem}

The proof of the theorem is given in the following, which applies
Lemma~\ref{lemma:k-cost}. This applied lemma is nontrivial and is proved below.

\begin{lemma}\label{lemma:k-cost}
	For any project $N$, its $k$-crashing plan (where $k \leq k_{max}$) costs at least $k$ times the cost of the optimum $1$-crashing plan.
\end{lemma}

\begin{proof}[Proof of Theorem~\ref{thm:upper bound}]
	For convenience, let $N_i= N(A_1\cup \ldots \cup A_i)$,
	for $0\leq i\leq k$. Note that $N_0=N$ is the original project.
	
	Fix $i$ in $\{1,\ldots,k\}$ in the following.
	By the algorithm,
	$A_{i}$ is the optimum $1$-crashing plan of $N_{i-1}$.
	Using Lemma~\ref{lemma:k-cost}, we know (1)
	any $(k+1-i)$-crashing plan of $N_{i-1}$ costs at least $(k+1-i)\cdot \cost(A_i)$.
	
	Let $Y=A_1\cup \ldots \cup A_{i-1}$ and $X=\OPT \setminus Y$; hence $X\cup Y\supseteq \OPT$.
	
	Observe that $N(Y\cup X)$ saves $k$ days comparing to $N$,
	because $Y\cup X\supseteq \OPT$ is $k$-crashing,
	whereas $N(Y)=N_{i-1}$ saves $i-1$ days comparing to $N$.
	So, $N(Y\cup X)$ saves $k-(i-1)$ days comparing to $N(Y)$,
	which means (2) $X$ is a $(k-i+1)$-crashing plan of $N(Y)=N_{i-1}$.
	Combining (1) and (2), $\cost(X)\geq (k+1-i)\cost(A_i)$.
	
	Further since $\cost(X)\leq \cost(\OPT)$ (as $X\subseteq \OPT$),
	we obtain a relation $(k-i+1)\cost(A_i)\leq \cost(\OPT)$.	Therefore,
	\begin{equation*}		
		\cost(G)=\sum_{i=1}^{k}\cost(A_i) \leq \sum_{i=1}^{k}\frac{1}{k-i+1} \cost(\OPT).
	\end{equation*}	
\end{proof}

The \emph{critical graph} of network $H$, denoted by $H^*$, is formed by all the critical edges of $H$;
all the edges not critical are removed in $H^*$.

Before presenting the proof to the key lemma (Lemma~\ref{lemma:k-cost}),
we shall briefly explain how do we find the optimum $1$-crashing plan of some project $H$ (e.g., the accelerated project $N(G)$ in Algorithm~\ref{alg:greedy-crashing}).
First, compute the critical graph $H^*$ and define the capacity of $e_i$ in $H^*$ by $c_i$ if $e_i$ in $H$ can still be shortend (i.e., its length is more than $a_i$) and otherwise define the capacity of $e_i$ in $H^*$ to be $\infty$. 
Then we compute the minimum $s-t$ cut of $H^*$ (using the max-flow algorithm \cite{flow-survey}), and this cut gives the optimum $1$-crashing plan of $H$.

\subsection{proof (part I)}\label{sect:bound}

\begin{proposition}\label{prop:cut}
	A $k$-crashing plan $X$ of $N$ contains a cut of $N^*$.
\end{proposition}

\begin{proof}
	Because $X$ is $k$-crashing, each critical path of $N$ will be shortened in $N(X)$ and that means it contains an edge of $X$.
	Further since paths of $N^*$ are critical paths of $N$ (which simply follows from the definition of $N^*$), each path of $N^*$ contains an edge in $X$.
	
	As a consequence, after removing the edges in $X$ that belong to $N^*$, we disconnect source $s$ and sink $t$ in $N^*$.
	Now, let $S$ denote the vertices of $N^*$ that can still be reached from $s$ after the remove, and let $T$ denote the remaining part.
	Observe that all edges from $S$ to $T$ in $N^*$, which forms a cut of $N^*$, belong to $X$.
\end{proof}

\newcommand{\mincut}{\mathsf{mincut}}
When $X$ contains at least one cut of network $H$, let $\mincut(H,X)$ be the minimum cut of $H$ among all cuts of $H$ that belong to $X$.

Recall that $d(H)$ is the duration of network $H$.

\medskip In the following, suppose $X$ is a $k$-crashing plan of $N$. We introduce a decomposition of $X$
which is crucial to our proof.

First, define
\begin{equation}
	\left\{\begin{array}{ccc}
		N_1&=&N,\\
		X_1&=&X,\\
		C_1&=& \mincut(N_1^*,X_1).
	\end{array}\right.
\end{equation}

\begin{figure}[h]
	\centering	\includegraphics[width=\linewidth]{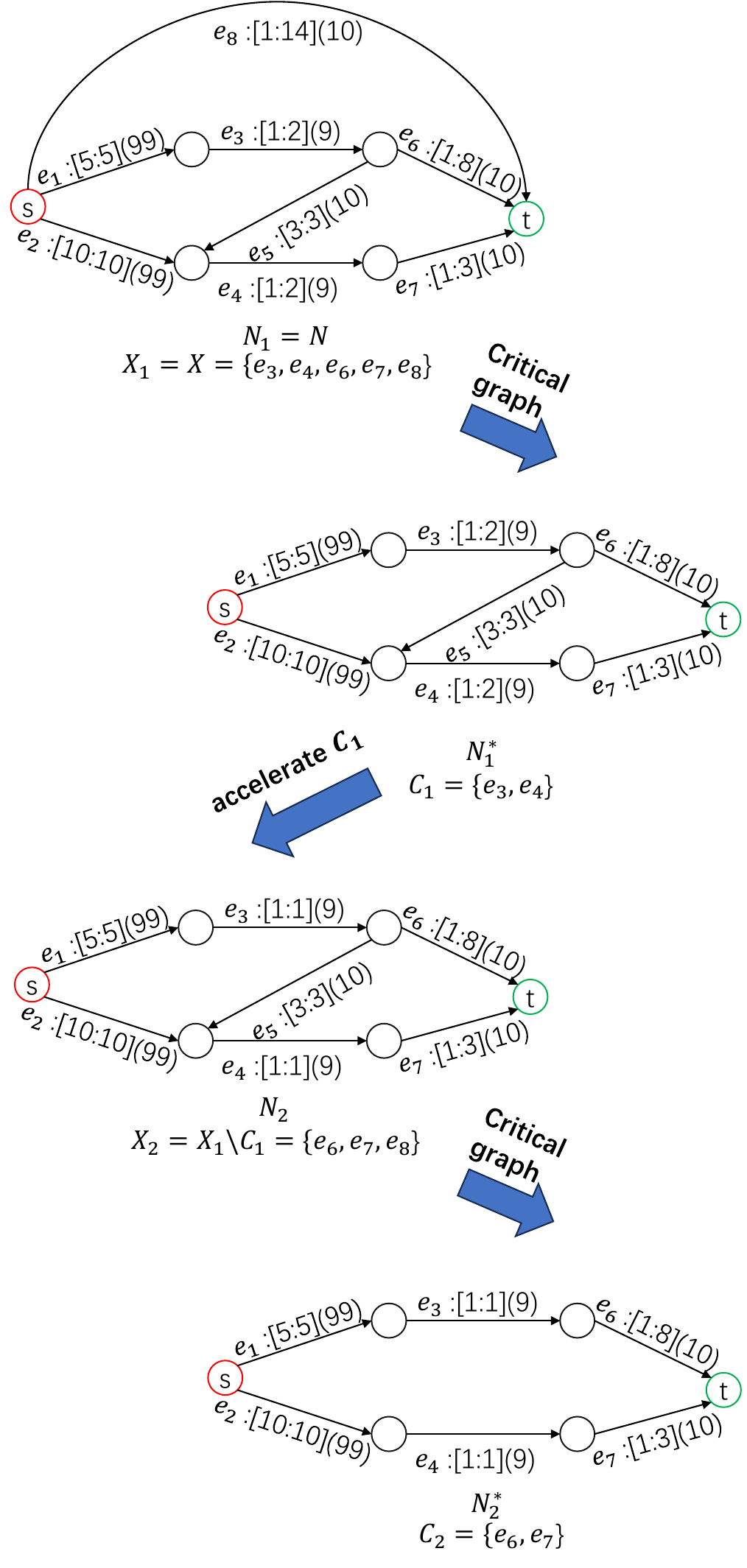}
	\caption{An example of the construction of $N_i,X_i,C_i$.} \label{fig-7}
\end{figure}

(Because $X_1=X$ is $k$-crashing, applying Proposition~\ref{prop:cut},
$X_1$ contains at least one cut of $N^*_1$. It follows that $C_1$ is well-defined.)

Next, for $1<i\leq k$, define
\begin{equation}\label{eqn:C}
	\left\{\begin{array}{ccc}
		N_{i}&=& N_{i-1}^*(C_{i-1}),\\
		X_{i}&=&X_{i-1}\setminus C_{i-1}.\\
		C_i &=& \mincut(N^*_i,X_i).
	\end{array}\right.
\end{equation}

See Figure~\ref{fig-7} for an example.

\begin{proposition}
	For $1< i\leq k$, it holds that
	
	1. $d(N_{i}) = d(N_{i-1}) - 1$ (namely, $d({N_i})=d(N)-i+1$).
	
	2. $X_i$ contains a cut of $N^*_i$ (and thus $C_i$ is well-defined).
\end{proposition}

\begin{proof}
	1. Because $C_{i-1}$ is a cut of $N^*_{i-1}$,
	set $C_{i-1}$ is a 1-crashing plan of $N^*_{i-1}$, which means
	$d(N^*_{i-1}(C_{i-1})) \leq d(N^*_{i-1})-1$.
	
	Moreover, it cannot hold that
	$d(N^*_{i-1}(C_{i-1})) \leq d(N^*_{i-1})-2$.
	Otherwise, cancel one edge of $C_{i-1}$ and
	it still holds that
	$d(N^*_{i-1}(C_{i-1})) \leq d(N^*_{i-1})-1$,
	which means $C_{i-1}$ is at least $1$-crashing for $N^*_{i-1}$,
	and thus contains a cut of $N^*_{i-1}$ (by Proposition~\ref{prop:cut}).
	This contradicts the assumption that $C_{i-1}$ is the minimum cut.
	
	Therefore,  $d(N_i)=d(N^*_{i-1}(C_{i-1})) = d(N^*_{i-1})-1 = d(N_{i-1})-1$.
	
	\medskip 2.
	According to Proposition~\ref{prop:cut}, it is sufficient to prove that $X_i$ is a 1-crashing plan to $N_i$.
	
	Suppose to the opposite that $X_i$ is not 1-crashing to $N_i$.
	There exists a path $P$ in $N^*_i$ that is disjoint with $X_i$.
	Observe that
	
	(1) The length of $P$ in $N(X)$ is the original length of $P$ (in $N$) minus the number of edges in $X$ that fall in $P$.
	
	(2) The length of $P$ in $N^*_i$ is the original length of $P$ (in $N$) minus the number of edges in $(C_1 \cup \ldots \cup C_{i-1})$
	that fall in $P$.
	
	(3) The number of edges in $(C_1 \cup \ldots \cup C_{i-1})$ that fall in $P$
	equals the number of edges in $X$ that fall in $P$, because
	$X \setminus (C_1 \cup \ldots \cup C_{i-1})=X_i$ is disjoint with $P$.
	
	Together, the length of $P$ in $N(X)$ equals the length of $P$ in $N^*_i$,
	which equals $d(N^*_i)=d(N_i)= d(N)-i+1>d(N)-k$.
	This means $X$ is not a $k$-crashing plan of $P$, contradicting our assumption.
\end{proof}



The following lemma easily implies Lemma~\ref{lemma:k-cost}.

\begin{lemma}\label{lem:cost_i}
	$ \cost(C_i) \leq \cost(C_{i+1})$ for any $i~(1\leq i<k)$.
\end{lemma}

We show how to prove Lemma~\ref{lemma:k-cost} in the following.
The proof of Lemma~\ref{lem:cost_i} will be shown in the next subsection.

\begin{proof}[Proof of Lemma~\ref{lemma:k-cost}]
	Suppose $X$ is $k$-crashing to $N$.
	
	By Lemma~\ref{lem:cost_i}, we know $\cost(C_1) \leq \cost(C_i)~(1\leq i\leq k)$.
	
	Further since $\bigcup^k_{i=1}C_i \subseteq X$,
	\begin{equation*}
		k \cdot \cost(C_1) \leq \cost(\bigcup^k_{i=1}C_i)\leq \cost(X).
	\end{equation*}
	
	Because $C_1$ is the minimum cut of $N^*$ that is contained in $X$,
	whereas $A_1$ is the minimum cut of $N^*$ among all, $\cost(A_1) \leq \cost(C_1)$.
	To sum up, we have $k \cdot \cost(A_1) \leq \cost(X).$
\end{proof}

It is noteworthy to mention that $\bigcup^k_{i=1}C_i$ is not always equal to $X$
and $\bigcup^k_{i=1}C_i$ may not be $k$-crashing.

\subsection{proof (part II)}

Assume $i~(1\leq i<k)$ is fixed.
In the following we prove that $\cost(C_i)\leq \cost(C_{i+1})$, as stated in Lemma~\ref{lem:cost_i}.
Some additional notation shall be introduced here.
See Figure~\ref{fig-4} and Figure~\ref{fig-5}.

Assume the cut $C_i$ of $N^*_i$ divides the vertices of $N^*_i$ into two parts, $U_i,W_i$,
where $s\in U_i$ and $t\in W_i$. The edges of $N^*_i$ are divided into four parts:
1. $S_i$ -- the edges within $U_i$;
2. $T_i$ -- the edges within $W_i$;
3. $C_i$ -- the edges from $U_i$ to $W_i$;
4. $RC_i$ -- the edges from $W_i$ to $U_i$.

\begin{figure}[h]
	\includegraphics[width=0.65\linewidth]{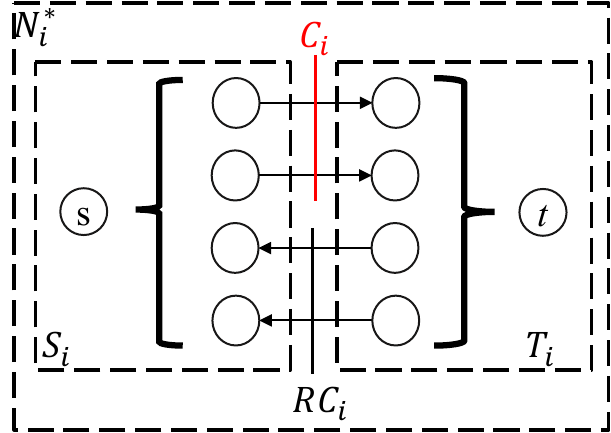}
	\centering
	\caption{Example of the definition of $S_i$, $T_i$, $RC_i$ and $C_i$}\label{fig-4}
\end{figure}

\begin{proposition}\label{prop:two-formulas}
	(1) $C_{i+1} \cap RC_i = \emptyset$ and (2) $C_i \subseteq N^*_{i+1}$.
\end{proposition}

\begin{proof}
	(1) Consider $N^*_{i}$. Any path involving $RC_i$ goes through $C_i$ at least twice.
	Such paths are shorten by $C_i$ by at least 2, and are thus excluded from $N_{i+1}$. However, $C_{i+1}$ are edges in $N^*_{i+1}$ and so are included in $N_{i+1}$. Together, $C_{i+1} \cap RC_i = \emptyset$.
	
	(2) Suppose there is an edge $e_i \in C_i$ and $e_i \notin N^*_{i+1}$. All paths in $N^*_i$ passing $e_i$ will be shortened by at least 2 after expediting $C_i$ to avoid becoming a critical path (which makes $e_i$ critical).
	If the shortening of $e_i$ is canceled, the paths can still be shortened by 1. So $C_i\setminus \{e_i\}$ still contains a cut to $N^*_i$, which violates the assumption that $C_i$ is the minimum cut, contradictory.
	So $C_i \subseteq N^*_{i+1}$.
\end{proof}

\begin{figure}[h]
	\centering	\includegraphics[width=\linewidth]{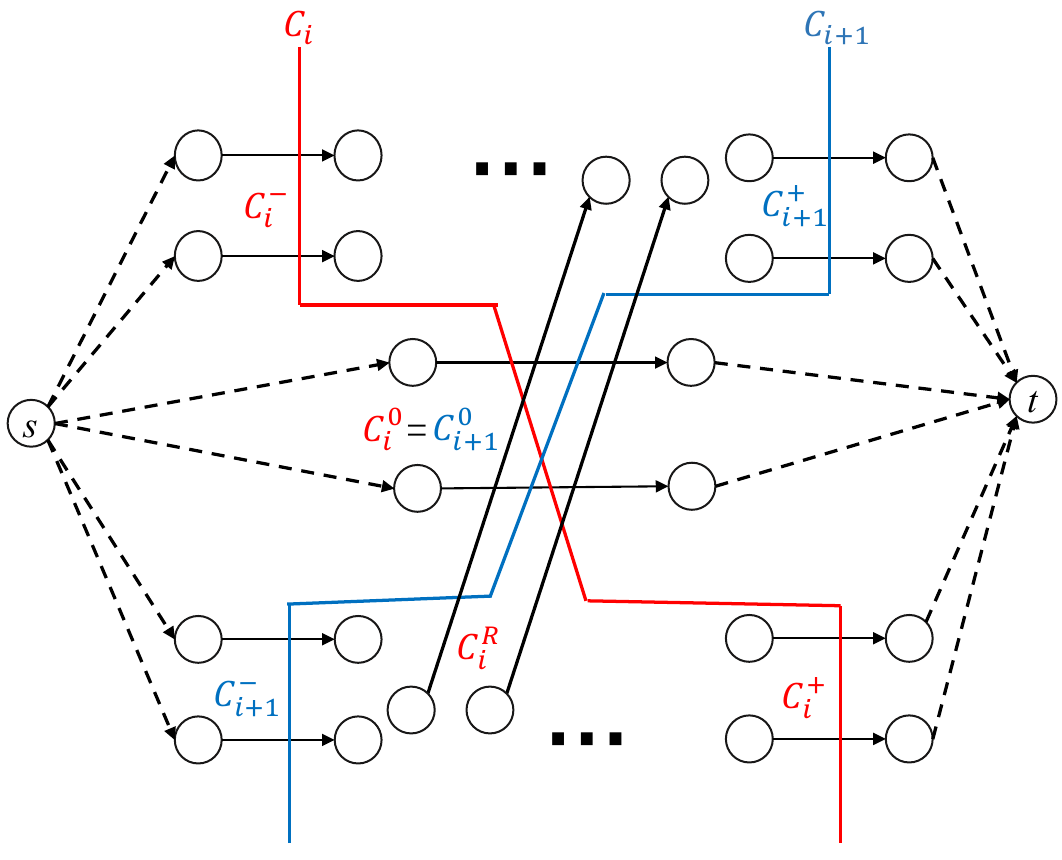}
	\caption{Key notation used in the proof of Lemma~\ref{lem:cost_i}.} \label{fig-5}
\end{figure}

Because $C_{i+1}$ is a subset of the edges of $N^*_{i+1}$,
and the edges of $N^*_{i+1}$ are also in $N^*_i$,
we see $C_{i+1} \subseteq T_i \cup S_i \cup C_i \cup RC_i$.
Further since $C_{i+1} \cap RC_i = \emptyset$ (proposition~\ref{prop:two-formulas}),
set $C_{i+1}$ consists of three disjoint parts:
\begin{equation}\label{eqn:Ci+1}
	\left\{
	\begin{split}
		C_{i+1}^+ &= C_{i+1}\cap T_i;\\
		C_{i+1}^0 &= C_{i+1}\cap C_i;\\
		C_{i+1}^- &= C_{i+1}\cap S_i.
	\end{split}
	\right.
\end{equation}

Due to $C_i \subseteq N^*_{i+1}$ (proposition~\ref{prop:two-formulas}), set $C_{i}$ consists of four disjoint parts:
\begin{equation}\label{eqn:Ci}
	\left\{
	\begin{split}
		C_i^+ &= C_i\cap T_{i+1}\\
		C_i^0 &= C_{i}\cap C_{i+1}\\
		C_i^- &= C_{i} \cap S_{i+1}\\
		C_i^R &= C_{i} \cap RC_{i+1}
	\end{split}
	\right.
\end{equation}

See Figure~\ref{fig-5} for an illustration. Note that $C_i^0=C_{i+1}^0$.

\begin{proposition}\label{prop:contain cut}~
	
	1. $C_{i+1}^+ \cup C_i^0 \cup C_i^+ $ contains a cut of $N^*_i$.
	
	2. $C_{i+1}^- \cup C_i^0 \cup C_i^- $ contains a cut of $N^*_i$.
\end{proposition}

\begin{proof}
	To show that $C_{i+1}^+ \cup C_i^0 \cup C_i^+ $ contains a cut of $N^*_i$,
	it is sufficient to prove that any path $P$ in $N^*_i$ goes through $C_{i+1}^+ \cup C_i^0 \cup C_i^+ $.
	
	Assume that $P$ is disjoint with $C_i^+ \cup C_i^0$, otherwise it is trivial.
	We shall prove that $P$ goes through $C_{i+1}^+$.
	
	Clearly, $P$ goes through $C_i = C_i^+ \cup C_i^0 \cup C_i^- \cup C_i^R$, a cut of $N^*_i$.
	Therefore, $P$ goes through $C_i^-\cup C_i^R$. See Figure~\ref{fig-6}.
	
	Take the last edge in $C_i^-\cup C_i^R$ that $P$ goes through; denoted by $e_a$ with endpoint $a$. Denote the part of $P$ after $e_a$ by $P^+$ ($P^+\cap C_i = \emptyset$).
	
	We now claim that
	
	(1) $P^+ \subseteq T_i$.
	
	(2) In $N^*_i$, there must be a path $P^-\subseteq S_{i}$ from the source to $e_a$ without passing $C_i$.($P^-\cap C_i = \emptyset$)
	
	(3) $a \in U_{i+1}$.
	
	Since $P$ is disjoint with $C_i^+ \cup C_i^0$ and $P^+$ has gone through  $C_i^-\cup C_i^R$, we obtain (1).
	
	If (2) does not hold, then all paths in $N^*_i$ that pass through $e_a$ have pass through $C_i$ already. In this case, $C_i\setminus e_a$ also contains a cut of $N^*_i$, which contradicts our definition of minimum cut $C_i$. Thus we have (2).
	
	By the definition (\ref{eqn:Ci+1}), any edge of $C_i^-\cup C_i^R$ ends at a vertex of $U_{i+1}$. Since $a$ is the endpoint of $e_a \in C_i^-\cup C_i^R$, we have (3).

	By (2), we can obtain a $P^-$ from $s$ to $e_a$. Concatenate $P^-,e_a,P^+$, we obtain a path $P'$ in $N^*_i$.
	Since $(P^-\cup P^+) \cap C_i = \emptyset$ and $e_a \in C_i$, $P'$ only goes through  $C_i^- \cup C_i^R$ once. So $P'$ is only shortened by 1 and still critical after expediting $C_i$.
	Therefore, $P'$ exists in $N^*_{i+1}$ .
	
	According to (3), we known that $P^+ \in P'$ starts with $a\in U_{i+1}$ and ends at the sink in $W_{i+1}$. Thus it must go though the cut $C_{i+1}$.
	
	According to (1) and definition (\ref{eqn:Ci+1}),
	path $P^+$ (which is a subset of $T_i$ due to (1)) can only go through $C_{i+1}^+ \subset T_{i}$.
	
	Since $P^+ \subset P$, we have $P$ goes through $C_{i+1}^+$. So any path $P$ in $N^*_i$ goes through $C_{i+1}^+ \cup C_i^0 \cup C_i^+ $.
	
	Therefore, $C_{i+1}^+ \cup C_i^0 \cup C_i^+ $ contains a cut of $N^*_i$.
	
	\begin{figure}[h]
		\includegraphics[width=\linewidth]{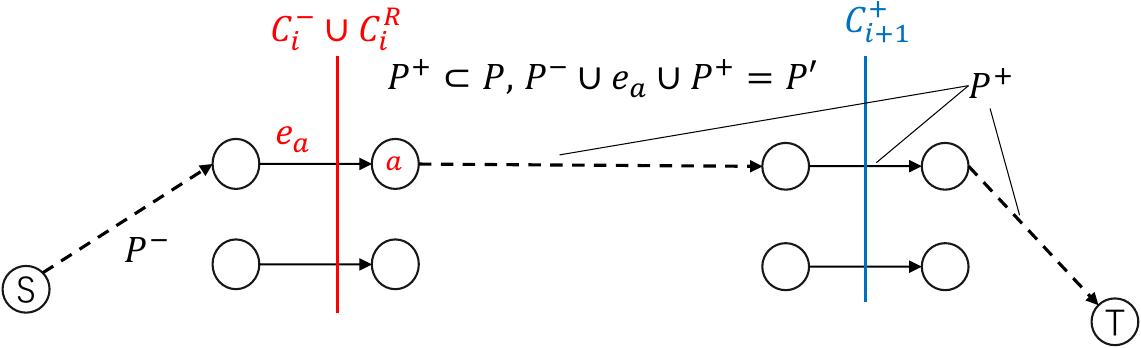}
		\centering
		\caption{Construction of $P'$ in the proof of Proposition~\ref{prop:contain cut}~part 1.} \label{fig-6}
	\end{figure}
	
	
	\medskip To show that $C_{i+1}^- \cup C_i^0 \cup C_i^- $ contains a cut of $N^*_i$,
	it is sufficient to prove that any path $P$ in $N^*_i$ goes through $C_{i+1}^- \cup C_i^0 \cup C_i^- $.
	
	Assume that $P$ is disjoint with $C_i^- \cup C_i^0$, otherwise it is trivial.
	We shall prove that $P$ goes through $C_{i+1}^-$.
	
	Clearly, $P$ goes through $C_i = C_i^+ \cup C_i^0 \cup C_i^- \cup C_i^R$, a cut of $N^*_i$.
	Therefore, $P$ goes through $C_i^+\cup C_i^R$. See Figure~\ref{fig-3}.
	
	Take the first edge in $C_i^+\cup C_i^R$ that $P$ goes through; denoted by $e_b$ with start point $b$. Denote the part of $P$ before $e_b$ by $P^-$ ($P^-\cap C_i = \emptyset$).
	
	We now claim that
	
	(1) $P^- \subseteq S_i$.
	
	(2) In $N^*_i$, there must be a path $P^+\subseteq T_{i}$ from $e_b$ to the sink without passing $C_i$.($P^+\cap C_i = \emptyset$)
	
	(3) $b \in W_{i+1}$.
	
	Since $P$ is disjoint with $C_i^- \cup C_i^0$ and $P^-$ hasn't reach $C_i^+\cup C_i^R$, we obtain (1).
	
	If (2) does not hold, then all paths in $N^*_i$ that pass through $e_b$ will pass through $C_i$ anyway. In this case, $C_i\setminus e_b$ also contains cut of $N^*_i$, which contradicts our definition of minimum cut $C_i$. Thus we have (2).
	
	By the definition (\ref{eqn:Ci+1}), any edge of $C_i^+\cup C_i^R$ starts with a vertex of $W_{i+1}$. Since $b$ is the start of $e_b \in C_i^+\cup C_i^R$, we have (3).

	By (2), we can obtain a $P^+$ from $e_b$ to the sink. Concatenate $P^-,e_b,P^+$, we obtain a path $P'$ in $N^*_i$.
	Since $(P^-\cup P^+) \cap C_i = \emptyset$ and $e_b \in C_i$, $P'$ only goes through  $C_i^+ \cup C_i^R$ once. So $P'$ is only shortened by 1 and still critical after expediting $C_i$.
	Therefore, $P'$ exists in $N^*_{i+1}$ .
	
	According to (3), we known that $P^- \in P'$ starts from the source and ends at $b\in W_{i+1}$. Thus it must go though the cut $C_{i+1}$.
	
	According to (1) and definition (\ref{eqn:Ci+1}),
	path $P^-$ (which is a subset of $S_i$ due to (1)) can only go through $C_{i+1}^- \subset S_{i}$.
	
	Since $P^- \subset P$, we have $P$ goes through $C_{i+1}^-$. So any path $P$ in $N^*_i$ goes through $C_{i+1}^- \cup C_i^0 \cup C_i^- $.
	
	Therefore, $C_{i+1}^- \cup C_i^0 \cup C_i^- $ contains a cut of $N^*_i$.
	
\end{proof}

\begin{figure}[h]
	\includegraphics[width=\linewidth]{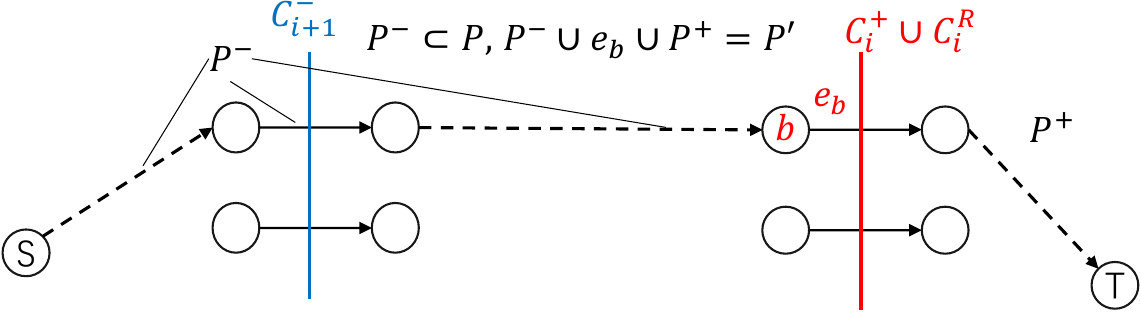}
	\centering
	\caption{Construction of $P'$ in the proof of Proposition~\ref{prop:contain cut}~part 2.} \label{fig-3}
\end{figure}

We are ready to prove Lemma~\ref{lem:cost_i}.

\begin{proof}[Proof of Lemma~\ref{lem:cost_i}]
	According to proposition~\ref{prop:contain cut},
	$C_{i+1}^+ \cup C_i^0 \cup C_i^+ $ and $C_i^- \cup C_i^0 \cup C_{i+1}^- $ each contains a cut of $N^*_i$.
	Notice that  $((C_{i+1}^+ \cup C_i^0 \cup C_i^+)\cup(C_i^- \cup C_i^0 \cup C_{i+1}^-)\cup C_i^R) = (C_i \cup C_{i+1}) \subset X_i$, so the mentioned two cuts are in $X_i$.
	Further since $C_i$ is the minimum cut of $N^*_i$ in $X_i$. We obtain
	\begin{equation*}
		\begin{split}
			\cost(C_i) &= \cost(C_i^+ \cup C_i^0 \cup C_i^- \cup C_i^R)\leq \cost(C_{i+1}^+ \cup C_i^0 \cup C_i^+)\\
			\cost(C_i) &=\cost(C_i^+ \cup C_i^0 \cup C_i^- \cup C_i^R) \leq \cost(C_{i+1}^- \cup C_i^0 \cup C_i^-)
		\end{split}
	\end{equation*}
	
	By adding the inequalities above (and note that $C_{i+1}^0 = C_i^0 = C_i \cap C_{i+1}$), we have
	\begin{equation*}
		\begin{split}
			2\cost(C_i^+ \cup C_i^0 \cup C_i^-\cup C_i^R) \\
			\leq \cost(C_{i+1}^+ \cup C_i^0 \cup C_i^+) + \cost(C_i^- \cup C_{i+1}^0 \cup C_{i+1}^-)
		\end{split}
	\end{equation*}
	Equivalently,
	\begin{equation*}
		\begin{array}{cc}
			& 2\cost(C_i^+ \cup C_i^0 \cup C_i^-) + 2\cost(C_i^R)\\
			\leq & \cost(C_i^- \cup C_i^0 \cup C_i^+ \cup C_{i+1}^+ \cup C_{i+1}^0 \cup C_{i+1}^-).
		\end{array}
	\end{equation*}
	
	By removing one piece of $\cost(C_i^- \cup C_i^0 \cup C_i^+)$ from both side,
	\begin{equation*}
		\cost(C_i) + \cost(C_i^R) \leq \cost(C_{i+1}^+ \cup C_{i+1}^0 \cup C_{i+1}^-)=\cost(C_{i+1})
	\end{equation*}
	
	Therefore, $\cost(C_i)\leq \cost(C_{i+1})$.
\end{proof}

Recall the ``convex case'' mentioned in subsection~\ref{subsect:formal-description},
in which shortening an edge becomes more and more expensive.
We claim that Algorithm~\ref{alg:greedy-crashing}
finds a solution with the approximation ratio $\frac{1}{1}+\ldots+\frac{1}{k}$ even for this convex case.
Essentially, the proof does not need any change. 
We leave the readers to verify this by themselves. 
(Hint: the value of the crashing cost is only applied in Lemma~\ref{lem:cost_i}, so all the other lemmas or propositions remain correct.)

\subsection{Counter-example of Algorithm~\ref{alg:greedy-crashing}}\label{subsect:counter-example}

The greedy algorithm given in Algorithm~\ref{alg:greedy-crashing} does not always find an optimal $k$-crashing plan.
Here we show an example.

\begin{figure}[h]
	\includegraphics[width=\linewidth]{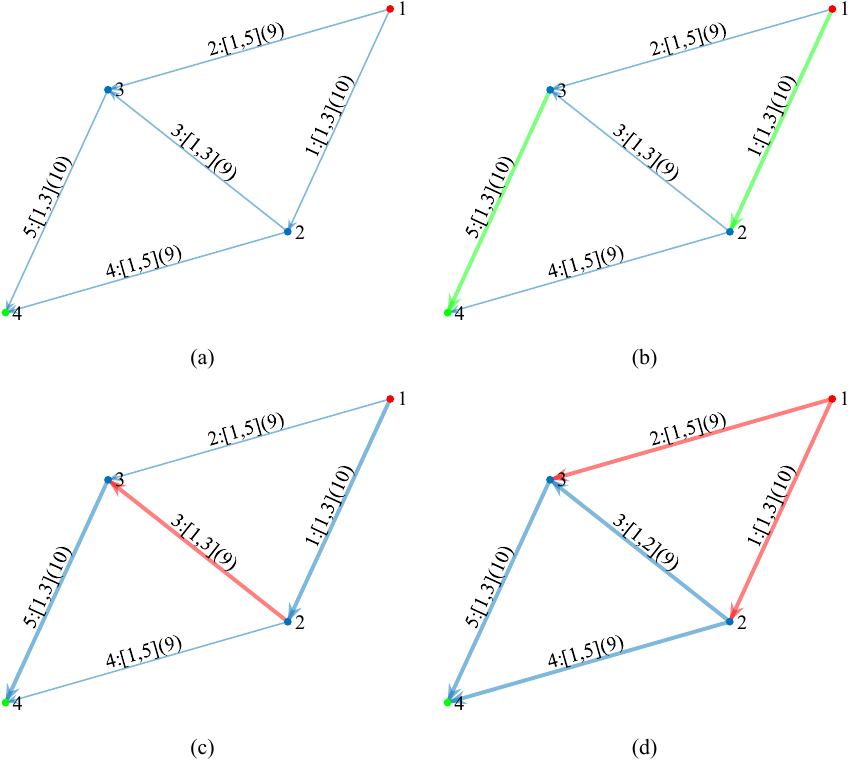}
	\centering
	\caption{An example with 5 jobs.
		The parameters $a_i,b_i,c_i$ of job $j_i$ are shown as a label $i:[a_i, b_i](c_i)$ in the graph. The source $s$ equals 1, whereas the sink $t$ equals 4.} \label{fig-2}
\end{figure}

The network is as shown in Figure~\ref{fig-2}~(a), and we consider $k=2$.
The critical path of this network has length $9$.
The unique critical path consists of jobs $j_1$, $j_3$, and $j_5$.

The greedy algorithm expedites job $j_3$ for one day in the first iteration; see Figure~\ref{fig-2}~(c).
It further expedites jobs $j_1$ and $j_2$ for one day in the second iteration; see Figure~\ref{fig-2}~(d). The total cost is $9+9+10=28$.

The optimal $k$-crashing plan is to expedite jobs $j_1$ and $j_5$, as shown in Figure~\ref{fig-2}~(b), which costs $10+10=20$.

\section{The $k$-LIS problem}\label{sec:k-LIS}

Recall that the $k$-LIS problem aims to find $k$ disjoint increasing subsequences among a given sequence $(a_1,a_2,\ldots,a_n)$,  so that the total length of the $k$ subsequences is maximized.
The following greedy algorithm can be used to find a solution for the problem.

\begin{algorithm}[hbpt]
	\caption{Greedy algorithm for $k$-LIS problem}\label{alg:greedy}
	\textbf{Input:} A sequence of number $\omega$.\\
	\textbf{Output:} $k$ increasing subsequences of $\omega$.
	
	\begin{algorithmic}
		\For{$i =1$ to $k$}
		\State Find the longest increasing subsequence $A_i$ of $\omega$;
		\State Output $A_i$ and remove $A_i$ from $\omega$;
		\EndFor
	\end{algorithmic}
\end{algorithm}

This algorithm does not guarantee an optimal result for $k>1$. For example, consider $ \omega = (3,4,5,8,9,1,6,7,8,9), k = 2$. We have the optimal result as $\{(3,4,5,8,9),(1,6,7,8,9)\}$ with the total length of $10$.
Yet the greedy algorithm finds $\{(3,4,5,6,7,8,9),(8,9)\}$ with a total length of $9$, which is not optimal.

\subsection{A lower Bound of the Greedy Algorithm }

We show that Algorithm~\ref{alg:greedy}
is an $(1-\frac{1}{e})$-approximation algorithm in the following.
Denote by $|\alpha|$ the length of any sequence $\alpha$.
Moreover, for any set $\mathcal{S}$ of (nonoverlapping) subsequences, let $|\mathcal{S}|$ denote the total length of subsequences in $\mathcal{S}$.

Let $L_1,\ldots,L_k$ denote the nonoverlapping subsequences given by the optimal solution of the $k$-LIS problem. Assume that $|L_1|\geq \ldots |L_k|$ without loss of generality.
Let $A_1,\ldots,A_k$ denote the nonoverlapping subsequences found by the greedy algorithm (Algorithm~\ref{alg:greedy}).
To be clear, $A_i$ refers to the subsequence found at the $i$-th step.

For any $x$ in $\{0,\ldots,k\}$, define $r_x(L_i) = L_i\setminus \{A_1\cup\ldots\cup A_x\}$,
which is the remaining part of $L_i$ after removing those elements selected in the first $x$ rounds of the greedy algorithm.

\medskip We start with the case $k = 2$, and we first bound $\sum_{i=1}^{k} |r_1(L_i)|$.

\begin{equation}\label{eqn:bound_r}
	|r_1(L_1)| + |r_1(L_2)| \geq |L_1| + |L_2| - |A_1|
\end{equation}

\begin{proof}[Proof of (\ref{eqn:bound_r})]
	Following the definition of $r_1(L_1)$ and $r_1(L_2)$, we have
	$$L_1 \cup L_2 \subseteq r_1(L_1) \cup r_1(L_2) \cup A_1,$$
	and thus
	$$|L_1 \cup L_2| \leq |r_1(L_1) \cup r_1(L_2) \cup A_1|.$$
	
	Note that the left side equals $|L_1|+|L_2|$,
	whereas the right side equals $|r_1(L_1)| + |r_1(L_2)|+|A_1|$,
	because of the disjointness of $L_1,L_2$ and that of $r_1(L_1),r_1(L_2),A_1$.
	Therefore, we obtain (\ref{eqn:bound_r}).
\end{proof}

Next, we point out two relations between $|A_i|$ and $|L_i|$.
\begin{eqnarray}
	2|A_1| & \geq & |L_1|+|L_2| \label{eqn:relation_A_L-1}\\
	2|A_2| + |A_1| & \geq & |L_1| + |L_2| \label{eqn:relation_A_L-2}
\end{eqnarray}

Summing up (\ref{eqn:relation_A_L-1}) divided by 4 with (\ref{eqn:relation_A_L-2}) divided by 2,
$|A_1| + |A_2| \geq \frac{3}{4}(|L_1|+|L_2|)$. So, Algorithm~\ref{alg:greedy} is a $\frac{3}{4}$-approximation for $k=2$.

\begin{proof}[Proof of (\ref{eqn:relation_A_L-1})]
	$2|A_1| = |A_1|+|A_1| \geq |L_1|+|L_2|$.
\end{proof}

\begin{proof}[Proof of (\ref{eqn:relation_A_L-2})]
	At the moment where the greedy algorithm is about to select $A_2$ (after chosen $A_1$),
	sequence $r_1(L_1)$ and $r_1(L_2)$ are both available for being selected by the greedy algorithm, hence
	\begin{equation}\label{eqn:relation-a-r}
		|A_2|\geq|r_1(L_1)| \text{ and }|A_2|\geq|r_1(L_2)|.
	\end{equation}
	
	Summing them up, $2|A_2| \geq |r_1(L_1)| + |r_1(L_2)|$.
	Combining this with (\ref{eqn:bound_r}), we have
	$2|A_2|	 \geq |L_1| + |L_2| - |A_1|$, i.e., (\ref{eqn:relation_A_L-2}) holds.
\end{proof}


We extend the above process to the case $k>2$ in the next.
Denote $\OPT=\{L_1,\ldots,L_k\}$,
and $g_x=|A_1|+\ldots+|A_x|$ for $x\leq k$.

\begin{lemma}\label{lemma:Lower} For each $x~(1\leq x\leq k)$,
	\begin{eqnarray}
		\sum_{i=1}^{k} |r_x(L_i)|&\geq& |\OPT| - g_x, \label{eqn:bound_r_extended}\\
		k|A_x| + g_{x-1} &\geq& |\OPT|  \label{eqn:relation_A_L_extended}
	\end{eqnarray}	
\end{lemma}

Be aware that (\ref{eqn:bound_r_extended}) extends (\ref{eqn:bound_r}), whereas
(\ref{eqn:relation_A_L_extended}) extends (\ref{eqn:relation_A_L-1}) and (\ref{eqn:relation_A_L-2}).

\begin{proof}[Proof of (\ref{eqn:bound_r_extended})]
	Denote $B=\bigcup_{j=1}^x A_j$ for short in this proof.
	Recall that $L_i\setminus B$ is defined to be $r_x(L_i)$.
	Clearly, $\OPT \setminus B= (L_1\cup \ldots \cup L_k) \setminus B= r_x(L_1)\cup \ldots \cup r_x(L_k).$	
	As a corollary, $$\OPT \subseteq (\OPT \setminus B) \cup B=r_x(L_1)\cup \ldots \cup r_x(L_k) \cup B.$$
	It follows that	 $|\OPT| \leq \sum_{i=1}^{k} |r_x(L_i)| + g_x$. Hence (\ref{eqn:bound_r_extended}) holds.
\end{proof}

\begin{proof}[Proof of (\ref{eqn:relation_A_L_extended})]
	Similar to (\ref{eqn:relation-a-r}), we have
	\begin{equation}\label{eqn:relation-a-r-extended}
		|A_{x}|\geq |r_{x-1}(L_i)| ~(\text{for each }i\in\{1,\ldots,k\}).
	\end{equation}
	
	Applying (\ref{eqn:bound_r_extended}),
	$g_{x-1} \geq |\OPT| - \sum_{i=1}^{k} |r_{x-1}(L_i)|.$
	
	Together, $g_{x-1} \geq |\OPT| - k |A_x|$; namely, (\ref{eqn:relation_A_L_extended}) holds.
\end{proof}

\begin{theorem}
	For any $k\geq 1$, it holds that
	\begin{equation}\label{eqn:ratio}
		\frac{g_k}{|\OPT|}\geq 1-(\frac{k-1}{k})^k>1-\frac{1}{e}.
	\end{equation}
\end{theorem}

\begin{proof}
	We first claim that for each $x~(1\leq x\leq k)$,
	\begin{equation}\label{eqn:kG}
		k g_x  \geq \sum_{i=0}^{x-1} (\frac{k-1}{k})^i|\OPT|.
	\end{equation}
	
	It follows that
	\begin{equation}
		\begin{split}
			k g_k \geq& \sum_{i=0}^{k-1} (\frac{k-1}{k})^i|\OPT|\\
			=& \frac{(\frac{k-1}{k})^{k}-1}{\frac{k-1}{k}-1}|\OPT|
			= k\left(1-(\frac{k-1}{k})^{k}\right)|\OPT|,
		\end{split}		
	\end{equation}
	which implies the ratio in (\ref{eqn:ratio}).
	
	We prove the essential formula (\ref{eqn:kG}) by induction in the following.
	
	Note that $k g_1 = k|A_1| \geq |L_1|+\ldots+|L_k| = |\OPT|$, whereas
	$\sum_{i=0}^{0} (\frac{k-1}{k})^i=1$. So (\ref{eqn:kG}) holds for $x=1$.
	Next, assuming that (\ref{eqn:kG}) holds for $x=m-1$,
	we shall prove that it also holds for $x=m$.
	\begin{equation*}
		\begin{split}
			k g_m & =  k (|A_m| + g_{m-1}) = (k|A_m|+g_{m-1}) + \frac{k-1}{k}(kg_{m-1})\\
			&\geq |\OPT| + \frac{k-1}{k}\sum_{i=0}^{m-2} (\frac{k-1}{k})^i|\OPT|\\
			&= |\OPT| + \sum_{i=1}^{m-1} (\frac{k-1}{k})^i|\OPT|=\sum_{i=0}^{m-1} (\frac{k-1}{k})^i|\OPT|.
		\end{split}
	\end{equation*}
	Thus the formula (\ref{eqn:kG}) also holds for $x$, completing the proof.
\end{proof}

\subsection{Examples where $\frac{g_k}{|\OPT|}\approx \frac{3}{4}$ for any $k\geq 2$}

This subsection constructs an example, in which
\begin{equation}\label{equa:3/4}
	\frac{g_k}{|\OPT|}=\frac{1}{k^2}\lceil\frac{3}{4}k^2\rceil
\end{equation}

First, we define a $k\times k$ matrix $M$, whose rows are indexed with $1,\ldots,k$ from bottom to top,
and columns are indexed with $1,\ldots,k$ from left to right.
Each element in the $i$-th row has a value $i$.
\begin{table}[!ht]
	\centering
	\begin{tabular}{|c|c|c|} \hline 
		$k$ & $\ldots$ & $k$   \\ \hline
		$\vdots$ &	$\ddots$ &	$\vdots$ \\ \hline
		1 & $\ldots$ & 1 \\ \hline
	\end{tabular}
\end{table}

Next, build a sequence $\omega$ of length $k^2$ by concatenating the following diagonals of $M$ ---
$(M_{1,1})$,  $(M_{2,1}, M_{1,2})$, \ldots, $(M_{k,1},\ldots,M_{1,k})$, \ldots, $(M_{k,k-1},M_{k-1,k})$, $(M_{k,k})$.
Elements within each diagonal are enumerated from top left to bottom right;
and diagonals are enumerated from left to right; e.g. $\omega=(1,2,1,3,2,1,3,2,3)$ for $k=3$.

It is obvious that $\omega$ can be partitioned into $k$ subsequences, each of which is a copy of $(1,\ldots,k)$.
This means $|\OPT|=k^2$ when $\omega$ is the given input of the $k$-LIS problem.
In the following we show that the greedy algorithm may return with $g_k=\lceil\frac{3}{4}k^2\rceil$.

Let $D_{1-k},\ldots,D_{k-1}$ denote the $2k-1$ diagonals of $M$ parallel to the minor diagonal of $M$;
formally, $D_x~(1-k\leq x\leq k-1)$ contains $M_{i,j}$ if $j=i+x$.
We regard $D_x$ as a sequence rather than a set ---
the elements are listed from bottom left to top right;
e.g., $D_0=(M_{1,1},\ldots,M_{k,k})$.
In this way, each $D_x$ is a subsequence of $\omega$.

\begin{table}[!ht]
	\centering
	\begin{tabular}{|c|c|c|c|} \hline 
		$D_{1-k}$ & $\ldots$ & $D_{-1}$ & $D_0$   \\ \hline
		$\vdots$ & $\iddots$ & $\iddots$ &	$D_1$ \\ \hline
		$D_{-1}$ & $D_0$ & $\iddots$ &	$\vdots$ \\ \hline
		$D_0$ & $D_1$ & $\ldots$ & $D_{k-1}$ \\ \hline
	\end{tabular}
\end{table}

\begin{lemma}
	The greedy algorithm above may return
	the longest $k$ diagonals among $D_{1-k},\ldots, D_{k-1}$
	as its solution (namely, it returns $A_1=D_0$, $A_{2i}=D_i$, $A_{2i+1}=D_{-i}$),
	yielding a total length $\lceil\frac{3}{4}k^2\rceil$.
\end{lemma}

\begin{proof}
	Clearly, $D_0$ can be chosen as $A_1$, because $|D_0|=k$.
	
	Suppose the lemma holds for the first $2j-1$ steps $j\geq 1$. Let $i = 2j$. Currently, the $2j-1$ longest subdiagonals of the table have been removed, leaving 2 symmetrical triangular areas at the top left and bottom right of the table. Let $T_{1}$ denote the top left triangular area and $T_{2}$ denote the bottom right triangular area.
	
	\begin{table}[!ht]
		\centering
		\begin{tabular}{|c|c|c|c|} \hline 
			k & $\ldots$ & k & $\times$   \\ \hline
			$\vdots$ & $\ddots$ & $\times$ &	$k-j$ \\ \hline
			$j+1$ & $\times$ & $\ddots$ &	$\vdots$ \\ \hline
			$\times$ & 1 & $\ldots$ & 1 \\ \hline
		\end{tabular}
	\end{table}
	
	Note that there are 2 subdiagonals with the length of $k-j$, one is $\{j+1,\ldots,k\}$, the other is $\{1,\ldots,k-j\}$, which are the longest subdiagonals at the moment. Now we prove that there is no increasing subsequence left that is longer than $k-j$.
	
	Suppose we pick an increasing subsequence $S$ from $\omega_{k\times k}$, we discuss its maximum length case by case.
	
	(\romannumeral1)If $S$ starts in $T_{1}$, it starts from a number no smaller than $j+1$. Since we only have $k$ numbers to form a sequence, $S$ can only be at most $k-j$ long.
	
	(\romannumeral2)If $S$ ends in $T_2$, it ends with a number no bigger than $k-j$. Since we only have $k$ numbers to form a sequence, $S$ can only be at most $k-j$ long.
	
	(\romannumeral3)If $S$ starts in $T_{2}$ and ends in $T_1$, it must have two consecutive numbers $s_1s_2$, so that $s_1$ is in $T_{2}$ and $s_2$ is in $T_{1}$. Then, we know that for  $s_1$ in $T_{2}$, we have $s_1+j$ to the left of it in $\omega_{k\times k}$. By observing the table, we know that for all possible options in $T_1$, we have $s_2>s_1+j$. Since we only have $k$ numbers to form a sequence, so $|S|\leq k-j$.
	
	Therefore, for all possible increasing subsequence $S$, we have $|S|\leq k-j$. So we take the two subdiagonals with the length of $k-j$ as the $D_{\max}$ for step $2j$ and $2j+1$, and they are the longest increasing subsequences at that time.
	
	\medskip The total length of $|A_1|+\ldots+|A_k|$ can be calculated as follows.
	\begin{equation*}
		\left\{
		\begin{array}{ll}
			k+2\sum_{i=1}^{\frac{k-1}{2}}(k-i)=\frac{3k^2+1}{4}=\lceil\frac{3}{4}k^2\rceil, & k\text{ is odd},\\
			k+2\sum_{i=1}^{\frac{k-2}{2}}(k-i)+\frac{k}{2}=\frac{3k^2}{4}=\lceil\frac{3}{4}k^2\rceil & k\text{ is even}.
		\end{array}\right.
	\end{equation*}
\end{proof}

\section{Summary \& future work}

We have shown that simple greedy algorithms achieve pretty small approximation ratio
in $k$-LIS and $k$-crashing problems. And the analysis is non-trivial.

Hopefully, the techniques developed in this paper can be used for analyzing greedy algorithms of other related problems.

\smallskip We would like to end up this paper with one challenging problem: Can we prove a constant approximation ratio for Algorithm~\ref{alg:greedy-crashing}?

\appendix






\bibliographystyle{ACM-Reference-Format}
\bibliography{k-crashing}


\end{document}